\documentclass[a4paper,runningheads,USenglish,thm-restate]{llncs}

\newcommand{\LV}[1]{#1}
\newcommand{\SV}[1]{}

\usepackage[OT2,T1]{fontenc}
\usepackage[USenglish]{babel}

\usepackage{placeins}

\LV{\usepackage[appendix=inline]{apxproof}}
\SV{\usepackage{apxproof}}%
\SV{}

\newtheoremrep{theorem}{Theorem}

\SV{\usepackage{lineno}}
\SV{\linenumbers}

\usepackage{amsmath,amssymb,amsfonts}
\usepackage{graphicx}
\usepackage[colorinlistoftodos]{todonotes}
\usepackage{hyperref}
\usepackage{xspace}
\usepackage{comment}
\usepackage{cancel}
\usepackage{mathtools}
\usepackage[all]{xy}

\def \emptyword{\lambda}

\usepackage{todonotes}
\usepackage{url}

\renewcommand{\phi}{\varphi}

\newcommand{\cL}{\mathcal{L}}

\newcommand{\todoJS}[1]{\todo[inline]{Jana: #1}}

\title{Forbidden-Context \& Ordered Grammar Systems}\author{Henning Fernau \inst{1}\orcidID{0000-0002-4444-3220} \and\\ Lakshmanan Kuppusamy\inst{2}\orcidID{0000-0003-2358-905X} \and\\Jana Schulz\inst{3}\orcidID{0009-0006-3168-2202}}
\authorrunning{H. Fernau, L. Kuppusamy, J. Schulz}

\institute{Fachbereich 4 -- Abteilung Informatikwissenschaften, Universit\"at Trier, 54286 Trier, Germany, \email{fernau@uni-trier.de}
 \and School of Computer Science and Engineering, VIT University, Vellore - 632~014, India, \email{klakshma@vit.ac.in} \and 
 Institut für Informatik und Computational Science, Universit\"at Potsdam, 14476 Potsdam, Germany, \email{jana.schulz@uni-potsdam.de}}  

\date{March 2026}

\begin{document}

\maketitle

\begin{abstract}
In this paper, we consider combining the ideas of forbidden random context grammars as well as of ordered grammars with cooperating distributed grammar systems (CDGS). We focus on investigating their generative capacities. Both ideas can be added to CDGS in two ways: either having (e.g.) a strict order of the rules in each component, or having a strict order on the components. This leads to four different scenarios, only some of them have been addressed  in the literature before. While in the area of CDGS, many inclusions among language classes have been 
open questions for decades, the proposed addition of forbidden random context and ordered regulation variants leads to a clear picture which allows us to get down to only five different classes of languages well known from classical regulated rewriting. This way, we also solve some open problems from the literature.   
\end{abstract}
%

\setcounter{footnote}{0}

\section{Introduction}

Cooperating distributed grammar systems (CDGS) with context-free rules have been introduced as a formal model of blackboard approaches in artificial intelligence. Different so-called modes have been introduced that govern the interaction between the components of a CDGS. For instance, one could prescribe that each component performs exactly $k$ steps before writing the sentential form back on the board so that others could take over; the $=2$-mode was also recently studied in the context of parsing~\cite{BorFerVas2026}.
Although CDGS have been extensively studied, in particular in the 1990s, many questions concerning their generative capacities are still open today; the inclusion diagram shown on page 72 of the well-known monograph~\cite{Csuetal94all} on grammar systems indicates many of these questions. In particular, none of these classical variants is known to describe the class \textsf{RE} of recursively enumerable languages. 
One idea to understand these questions is to add some more power to the rules and see how the inclusions might change by doing so. We follow this line of research here.

There are several papers that study grammar systems equipped with rules that have context conditions.
In particular, Masopust studies in \cite{Mas2009} CDGS with context-free rules equipped with forbidden nonterminal sets and proves that in the $t$-mode of derivation, these CDGS are equivalent to programmed grammars with appearance checking, both when allowing or disallowing erasing productions; i.e., when allowing them, they characterize \textsf{RE}.
In \cite{CsuMasVas2009}, a similar study was undertaken for CDGS  with context-free rules equipped with permitting nonterminal sets and it was proved that in the $t$-mode of derivation, such systems are as powerful as permitting random context grammars. In \cite{CsuMasVas2009,GolMasMed2010,Mas2009}, also a so-called left-permitting respectively left-forbidding derivation was considered; generally speaking, this variant seems to be more powerful, but in this paper, we like to focus on the classical variants of random context rules as introduced in~\cite{Wal72}.
However, interestingly, in connection with left-permitting rules, in \cite{GolMasMed2010}, also other modes of derivation have been examined and it was shown that then, e.g., characterizations of \textsf{RE} can be obtained.
In~\cite{Csuetal94all}, the fourth chapter is dedicated to certain control mechanisms integrated into grammar systems.
In Subsection 4.2, context conditions (e.g., random context conditions) work as filters, either applied as entry or also as exit 
conditions concerning the work of one grammar component. In this filter interpretation, forbidden context has not been  studied as we  do in this paper, in combination with all possible derivation modes.
\LV{It should be mentioned that forbidden context has been studied also with other types of grammatical systems. For instance, in \cite{FreOsw2002} it was shown that certain membrane systems arrive at computational completeness by adding forbidden contexts. Similarly, in \cite{MedSve2003} it was shown that forbidden context added to E(T)0L systems characterizes \textsf{RE}.}

Moreover, we consider grammar systems with ordered context-free rules, taking forward the classical idea of ordered grammars, introduced in~\cite{Fri68}. For more information, we refer to \cite{DasPau85,DasPau89,Fer96a}. Ordered grammars are well-known to be equivalent to forbidden random context grammars; we will see that this equivalence partially carries over to CDGS: if erasing rules are permitted, they characterize \textsf{RE}, but not in all modes. 
 Mitrana \emph{et al.}~\cite{MitPauRoz94} studied CDGS with \emph{priorities}, which means that there exists an ordering (in the sense of ordered grammars) on the components of the grammar system. This is again a different approach, as we have an ordering per component on the sets of rules.

We survey the approach of this paper in Table~\ref{tab:conceptual-survey}. Notice that most of the cited papers do not consider all classical modes, or at least they do not fully characterize the defined language classes as we will do in this paper.

\begin{table}[bt]
    \centering
    \begin{tabular}{|l|l|l|}\hline
     & strict ordering    & forbidden random context \\\hline
 w.r.t. rules in each component      & Section~\ref{sec:ocdgs} & \cite{Mas2009}, Section~\ref{sec:RCCDGS} \\
 between components & \cite{MitPauRoz94}, Section~\ref{sec:PCDGS}  & \cite{Csuetal94all} (entry filters), Section~\ref{sec:CDGS-rc} 
 \\\hline
    \end{tabular}
    \caption{Organization of the concepts and plan of the paper.}
    \label{tab:conceptual-survey}
\end{table}

\smallskip
\noindent\underline{Conventions.} We usually employ standard notations, e.g., $[n]=\{1,\dots,n\}$ for a positive integer~$n$; this notation is sometimes extended as $[i\dots j]$ for integers $i,j$, denoting the set of all integers $k$ with $i\leq k\leq j$.
Moreover, if $B=\{b\}$ is a singleton set, we write $A-b$ instead of $A\setminus B$.

\section{Ordered Cooperating Distributed Grammar Systems}
\label{sec:ocdgs}

An \emph{ordered (context-free) grammar} (introduced in~\cite{Fri68}) is a construct $G=(N,\Sigma,P,>,S)$, where 
$(N,\Sigma,P,S)$ is a classical context-free grammar, containing rules of the form $A\to y$, and $>$ is a strict ordering on~$P$ (i.e., $>$ is transitive, asymmetric, and hence also irreflexive).
The derivation relation $\Rightarrow$ is defined as follows.
For $u,v\in (N\cup\Sigma)^*$, $u\Rightarrow v$ if 
\begin{itemize}
    \item $\exists A\to y\in P\, \exists x,z\in (N\cup\Sigma)^*: u=xAz, v=xyz$ and
    \item $\forall A'\to y'\in P:((A'\to y')>(A\to y))\implies |u|_{A'}=0$.
\end{itemize}
While the first condition is just the classical rewrite condition, the second one prohibits that any larger rule might have been applicable.
When we specify an ordered grammar, we often do this by writing, e.g., 
$$\{X\to X\mid X\in\Phi\}>\{Y\to Y'\mid Y\in \Psi\}>Z\to Z''\,,$$
meaning that for all $X\in\Phi$ and for all $Y\in \Psi$, we have $X\to X>Y\to Y'$ and $Y\to Y'>Z\to Z''$ and hence also $X\to X>Z\to Z''$ by transitivity. By allowing or disallowing erasing rules, we arrive at the language classes $\textsf{ORD}$ and $\textsf{ORD}^{-\emptyword}$ respectively. Even if we only consider languages that are subset of $\Sigma^+$, i.e., ignoring the empty word, it is an old open problem if the trivial inclusion $\textsf{ORD}^{-\emptyword}\subseteq \textsf{ORD}$ is strict or not.

\newcommand{\Modes}[1]{\ensuremath{\textit{Modes}_{\,\geq #1}}}

Let \Modes{i} denote $\{\leq k,=k,\geq k\mid k\in\mathbb N,k\geq i\}\cup\{*,t\}$. 
An \emph{ordered cooperating distributed grammar system}, or OCDGS for short,  is a construct $G=(N,\Sigma,(P_1,>_1),\dots,(P_n,>_n),S)$, where each component $G_i=(N,\Sigma,P_i,{>_i},S)$ is an ordered grammar.
Let $\Rightarrow_i^m$ refer to the derivation relation of $G_i$, with $m\in \Modes{1}$.
For $m=t$ (i.e., the $t$-mode of derivation), $u\Rightarrow^t_i v$ means $u\Rightarrow_i^*v$ and there is no $w$ such that $v\Rightarrow_i w$. The notations $\Rightarrow_i^{\leq k}$,  $\Rightarrow_i^{=k}$ and  $\Rightarrow_i^{\geq k}$ are referring to at most, exactly, and at least~$k$ derivation steps of $G_i$ (as usual). We say that a derivation step $u\Rightarrow_i^m v$ is \emph{unproductive} if $u= v$; if $u\neq v$, we also say that we have applied $G_i$ \emph{productively}. In the simulations that we will discuss in our paper, unproductive derivation steps can be safely ignored. 
If $G$ works in any of the modes~$m\in\Modes{1}$, 
then
the derivation relation $\Rightarrow_{G,m}$ is defined as $\bigcup_{i=1}^n\Rightarrow_i^m$.
Accordingly, $\cL(G,m)=\{w\in \Sigma^*\mid S\Rightarrow_{G,m}^*w\}$. 
When all orderings $>_i$ are empty, then an OCDGS is nothing else than a classical cooperating distributed grammar system, or CDGS for short.
We therefore arrive at the language families $\textsf{(O)CD}(m)$, depending on the mode of derivation. We will add $-\emptyword$ as a superscript to our notation when we disallow erasing productions.

\begin{example}\label{ex:a^2^n-OCDGS}
To showcase how OCDGS work, we will provide a short example. Consider the OCDGS~$G=(\{A,B,C\},\{a\}, (P_1, >_1), (P_2,>_2),(P_3,>_3), A)$ with the components $G_i=(\{A,B,C\},\{a\}, P_i, >_i, A)$ for $i\in [3]$. The language of $G$ in $t$-mode is $\cL(G,t)=\{a^{2^n} \mid n \geq 0 \}$.
$P_1=\{A \to B, A \to C\}$ with $>_1=\emptyset$, the second component~$G_2$ contains the ordered rules $C\to C > B\to AA$ and the third component~$G_3$ contains the ordered rules $B\to B > C\to a$. The unproductive rules $B \to B$ and $C \to C$ are used to block the components from deriving $C\to a$ and $B\to AA$ respectively when both $B$ and $C$ are present in the sentential form. A successful exemplary derivation of a word in $\cL(G,t)$ could look like this:
$$A \Rightarrow_{1}^t B \Rightarrow_{2}^t AA \Rightarrow_{1}^t BB \Rightarrow_{2}^t AAAA \Rightarrow_{1}^t CCCC \Rightarrow_{3}^t aaaa$$
When component~$G_1$ used both rules, i.e. when there are $B$ and $C$ in the sentential form, the derivation will not be successful:
$$A \Rightarrow_{1}^t B \Rightarrow_{2}^t AA \Rightarrow_{1}^t BC \Rightarrow_{2}^* BC$$
Since the rule $C\to C$ remains applicable after any number of derivation steps and $G$ works in $t$-mode, component~$G_2$ can never be left.
\qed
\end{example}




\noindent
For proving our next results, we need another concept from regulated rewriting.

A \emph{context-free graph-controlled grammar} is a construct 
\begin{equation*}
G_{C}=\left( N,\Sigma,\left( P,L_{\text{init}},L_{\text{fin}}\right) ,S\right) ;
\end{equation*}%
$N$ and $\Sigma$ are disjoint alphabets of nonterminal and terminal symbols,
respectively; $S\in N$ is the start symbol; $P$
is a finite set of rules $r$ of the form $\left( \ell\left( r\right) :p\left(
\ell\left( r\right) \right) ,\sigma \left( \ell\left( r\right) \right) ,\varphi
:\left( \ell\left( r\right) \right) \right) $, where $\ell\left( r\right) \in
\text{Lab}\left( G_{C}\right) $, $\text{Lab}\left( G_{C}\right) $ being a set of labels
associated to the rules in~$R$ by the injective labeling mapping~$\ell$, $p\left(
\ell\left( r\right) \right) $ is a context-free production over $\left( N\cup
\Sigma\right) ^{\ast }$, $\sigma \left( \ell\left( r\right) \right) \subseteq
\text{Lab}\left( G_{C}\right) $ is the \textit{success field} of the rule~$r$, and $\varphi \left( \ell\left( r\right) \right) $ is the \textit{failure field} of
the rule $r$; $L_{\text{init}}\subseteq \text{Lab}\left( G_{C}\right) $ is the set of
initial labels, and $L_{\text{fin}}\subseteq \text{Lab}\left( G_{C}\right) $ is the set of
final labels. For $\left( \ell(r):p\left( \ell\left( r\right) \right) ,\sigma
\left( \ell\left( r\right) \right) ,\varphi \left( \ell\left( r\right) \right)
\right) $ and $v,w\in \left( N\cup \Sigma\right) ^{\ast }$ we define $\left(
v,\ell\left( r\right) \right) \Longrightarrow _{G_{C}}\left( w,k\right) $ if
and only if

\begin{itemize}
\item \textbf{either} $p\left( \ell\left( r\right) \right) $ is applicable to $v $, the result of the application of the production $p(\ell(r))$ to $v$ is $w$, and $k\in \sigma \left( \ell\left( r\right) \right) $,

\item \textbf{or} $p\left( \ell\left( r\right) \right) $ is not applicable to $v $, $w=v$, and $k\in \varphi \left( \ell\left( r\right) \right) $.
\end{itemize}

\noindent The language generated by $G_{C}$ is
\begin{equation*}
\begin{array}{lll}
L\left( G_{C}\right){} = & \{w\in \Sigma^{\ast }\mid {}& \left. \left(
w_{0},\ell_{0}\right) \Longrightarrow _{G_{C}}\left( w_{1},\ell_{1}\right)
\Longrightarrow _{G_{C}}\cdots \left( w_{k},\ell_{k}\right) ,\ k\geq 1,\right.
\\ 
&  & \left. w_{j}\in \left( N\cup  \Sigma\right) ^{\ast }\ \mathrm{and}\ \ell_{j}\in
\text{Lab}\left( G_{C}\right) \ \mathrm{for}\ 0\leq j\leq k,\right. \\ 
&  & \left. w_{0}=S,\ w_{k}=w,\ \ell_{0}\in L_{\text{init}},\ \ell_{k}\in L_{\text{fin}}\right\} .%
\end{array}%
\end{equation*}%

Let $\textsf{GC}_{ac}$ and $\textsf{GC}_{ac}^{-\emptyword}$ denote the classes of languages that can be generated by graph-controlled grammars (with appearance checking) allowing and disallowing erasing rules, respectively. It is known that $\textsf{GC}_{ac}=\textsf{RE}$, and this characterization holds even for graph-controlled grammars with only two nonterminals, see~\cite{Feretal07}. Furthermore, the strict inclusion of $\textsf{GC}_{ac}^{-\emptyword}$ in the class of context-sensitive languages $\textsf{CS}$ is known, see~\cite{Ros69} in the related formalism of programmed grammars. 



\begin{theorem}\label{thm:GCac-characterization-by-OCD-all-modes}
For any mode $m\in\{\geq k, =k\mid k\geq 2\}
$, $\textsf{GC}_{ac}^{\,(-\emptyword)}=\textsf{OCD}(m)^{(-\emptyword)}$.\LV{\footnote{This notation (used in many places in the paper) is giving two claims in compact notation: $\textsf{GC}_{ac}^{\,-\emptyword}=\textsf{OCD}(m)^{-\emptyword}$ and $\textsf{GC}_{ac}=\textsf{OCD}(m)$. In other words, one has to either consistently disallow erasing rules (on both sides of the equation) or allow them.}}
\end{theorem}

\begin{proof}

The following construction showing how to simulate graph-controlled grammars by OCDGS interestingly works for all modes $m \in \{=k, \geq k \mid k \geq 2\}$, as the reader can verify. 
To this end, let $G_{C}=\left( N_C,\Sigma,\left( P,L_{\text{init}},L_{\text{fin}}\right) , S_C\right)$ be some graph-controlled grammar with appearance checking, with $\text{Lab}(G_C)=\{\ell_i\mid i\in [\text{Lab}(G_C)]\}$. We are going to construct an equivalent  ordered CDGS $$G=(N,\Sigma,(P_0,>_0=\emptyset), (P_1,>_1), \dots, (P_{n},>_{n}),(P_{n+1},>_{n+1}) ,S)\,,$$ with $n=3|P|$, as follows.
Let $\widehat N_C=\{\widehat A\mid A\in N_C\}$, and similarly, $\widehat{\text{Lab}}(G_C)=\{\widehat\ell\mid \ell\in \text{Lab}(G_C)\}$. Let $N=N_C\cup\widehat N_C\cup \text{Lab}(G_C)\cup \widehat{\text{Lab}}(G_C)\cup\{S\}$.

The initialization is performed by the component $$P_0=\{S\to \ell_{\text{init}}S_C,\ell_{\text{init}}\to \ell_{\text{init}}\mid \ell_{\text{init}}\in L_{\text{init}}\}\,.$$
As this is the only occasion that we find a rule with left-hand side~$S$ in~$G$, we can only apply this component once in the beginning productively, and we also must apply this component in the beginning, resulting in the sentential form $\ell_{\text{init}}S_C$. As an invariant of our construction, we will maintain the presence of a symbol from $\text{Lab}(G_C)\cup \widehat{\text{Lab}}(G_C)$ prefixing any  sentential form $w\notin\Sigma^*$ with
$S\Rightarrow_{G,=2}^+w$, i.e., if $S\Rightarrow_{G,=2}^+w$, then $w\in\Sigma^*$ or $w\in (\text{Lab}(G_C)\cup \widehat{\text{Lab}}(G_C))\cdot (N_C\cup\widehat N_C\cup\Sigma)^*$.
In the following inductive argument, we discuss a sentential form $u=\ell v$ with $\ell\in (\text{Lab}(G_C)\cup \widehat{\text{Lab}}(G_C))$, $S\Rightarrow_G^+u$ and $(S_C,\ell_0)\Rightarrow_{G_C}^*(v,\ell)$ by induction hypothesis.

We terminate with the help of the component $$P_{n+1}:\{ X\to X\mid X\in N_C\cup \widehat N_C\}>_{n+1}\{\ell_{\text{fin}}\to\emptyword, \ell_{\text{fin}}\to\ell_{\text{fin}}\mid \ell_{\text{fin}}\in L_{\text{fin}}\}\,.$$
This component can only be applied productively to any sentential form~$u=\ell v$ if it does not contain any nonterminal from $N_C\cup\widehat N_C$. Even then, an application can be productive only if $\ell\in  L_{\text{fin}}$, in which case for $u\Rightarrow_{G,n+1}^{=2}$, the only productive generation yields $w=v\in \Sigma^*$. By definition of $\cL(G_C)$ and by our inductive hypothesis, also $w\in \cL(G_C)$ holds.

Now we describe how a rule $(\ell_i:A_i\to w_i,\sigma_i,\phi_i)$, $i\in [|\text{Lab}(G_C)|]$, is simulated. We can assume that $\ell_i\notin\phi_i$, as a second failure test is clearly unnecessary.
We define three simulating grammar components. As auxiliary set, define $\Lambda_{-i}=[|\text{Lab}(G_C)|]\setminus\{i\}$.
\begin{itemize}
    \item Failure case simulation by component $P_{k}$ for $k=3i$: $$\{\ell_j\to\ell_j\mid j\in\Lambda_{-i}\}>_{k}A_i\to A_i>_{k}\{\ell_i\to\ell'\mid \ell'\in \phi_i\}\,.$$
    In order to apply $P_k$ productively to $u=\ell v$, we must see $\ell=\ell_i$; also $A_i$ must not occur in~$v$. Hence, we will see $u=\ell v\Rightarrow_k\ell'v$ for some $\ell'\in \phi_i$. 
    As $\ell\neq\ell'\in\phi_i$, we are able to apply $P_k$ a second time, using $\ell'\to\ell'$. Hence, $u\Rightarrow_k^2\ell'v$, which faithfully simulates the derivation step $(v,\ell)\Rightarrow_{G_C}(v,\ell')$ of $G_C$ when applying the rule labeled $\ell_i$ in the absence of $A_i$ in~$v$.
    \item For the simulation of the success case, we use two components $P_k$ and $P_{k+1}$, with $k=3i-2$.\\
    $P_k:\{\ell_j\to\ell_j,\widehat\ell_j\to\widehat\ell_j\mid j\in \Lambda_{-i}\}\cup \{\widehat A\to\widehat A\mid A\in N_C\}$\\\phantom{XXXXXX}$>_k\ell_i\to\widehat\ell_i>_kA_i\to\widehat A_i$\,;\\
    $P_{k+1}: \{\widehat\ell_j\to\widehat\ell_j\mid j\in \Lambda_{-i}\}\cup\{\ell_j\to\ell_j\mid j\in [|\text{Lab}(G_C)|]\}$\\\phantom{XXXXXXX}$>_{k+1}\widehat A_i\to w_i>_{k+1}\{\widehat\ell_i\to \ell'\mid\ell'\in\sigma_i\}$\,.

    In order to apply $P_k$ productively on $u=\ell v$, we must see $\ell=\ell_i$, as for any other $\ell$, the first big set of rules in our ordering~$>_k$ will render the application unproductive. Hence, we observe $u\Rightarrow_{k}\widehat\ell_iv$. We have to make a second step with $P_k$. As $\widehat\ell_i$ is neither contained in the first big set of rules, the only possibility for a second rule application is to take $A_i\to\widehat A_i$; if $A_i$ does not occur in~$v$, the derivation will be stuck. If $A_i$ occurs in~$v$, one of its occurrences is turned into~$\widehat A_i$, which describes the resulting string~$v'$, i.e., $u\Rightarrow_{k}^{=2}\widehat\ell_iv'=u'$. Also notice that due to the required uniqueness of rule labelings, no other grammar component can derive productively on~$u$ but (possibly) $P_k$.

    Similarly, we can observe that the only grammar component that  can derive productively on~$u'$ is~$P_{k+1}$. On applying~$P_{k+1}$ once, we get the sentential form $\widehat\ell_iv''$, with $v''$ obtained from $v'$ by replacing the unique occurrence of $\widehat A_i$ in~$v'$ by~$w_i$.
    With similar arguments, we see that the only way to perform a second step of $P_{k+1}$ is to replace $\widehat\ell_i$ by some $\ell'\in\sigma_i$. Hence, $u'\Rightarrow_{k+1}^{=2}u''=\ell'v''$ is enforced. This clearly corresponds to $(v,\ell_i)\Rightarrow_{G_C}(v'',\ell_i)$ as a valid derivation step in~$G_C$, hence proving the inductive step, also for our invariant.
\end{itemize}
The claim that $\cL(G)\subseteq\cL(G_C)$ now follows by induction. Also the other inclusion direction can be read off from the previous argument and follows by another (much simpler) induction.

For the erasing case, since $\textsf{GC}_{\text{ac}}=\textsf{RE}$, the claim $\textsf{GC}_{ac}=\textsf{OCD}(=2)$ follows immediately.
The non-erasing case is therefore more interesting. In order to understand $\textsf{GC}_{ac}^{-\emptyword}\subseteq\textsf{OCD}(m)^{-\emptyword}$ completely, consder some $L\in \textsf{GC}_{ac}^{-\emptyword}$. As $\textsf{GC}_{ac}^{-\emptyword}$ is closed under left derivatives, we can decompose $L$ as $$L=\bigcup_{a\in\Sigma}\{a\}\delta_a(L)\cup F=\bigcup_{a\in\Sigma}\{a\}\{w\in \Sigma^+\mid aw\in L\}\cup F\,,$$ where $F\subseteq\Sigma$. Now, apply the construction given above to each graph-controlled grammar $G_{C,a}$ for $\delta_a(L)$ separately and replace the then only erasing rules of the form $\ell_{\text{fin}}\to\emptyword$ by $\ell_{\text{fin}}\to a$. Then, our argument gives a simulating ordered CDGS for  $\{a\}\delta_a(L)$. As $\textsf{OCD}(=2)$ can be easily seen to be closed under union and finite sets $F$ are trivially included in this class, the claim follows.\qed
\end{proof}

Notice that also in this simulation, the orderings that we create are of dimension two.
Hence, this also applies to the non-erasing case.
In the erasing case, we could further simplify the construction by replacing the two components that simulate success fields by the following single component:
$$\{\ell_j\to\ell_j\mid j\in \Lambda_{-i}\}>\ell_i\to\emptyword>\{A_i\to w_i\ell'\mid \ell'\in\sigma_i\}\,.$$

\begin{toappendix}
With the intuition that the modes from $\{\geq k, =k\mid k\geq 2\}$ are ``similar but slightly weaker than'' graph-controlled\LV{ (or equivalently, programmed or matrix)} grammars without appearance checking, the previous result can be seen as strengthening results from \cite{DasPau85} where it has been shown that these regulation mechanisms in combination with an ordering on the rules (enforcing to apply a ``largest rule'' from a success or failure field) are as powerful as appearance checking.
\end{toappendix}
\SV{With the intuition that the modes from $\{\geq k, =k\mid k\geq 2\}$ are ``similar but slightly weaker than'' graph-controlled\LV{ (or equivalently, programmed or matrix)} grammars without appearance checking, the previous result can be seen as strengthening results from \cite{DasPau85}.}

In an ordered grammar, a rule of the form $A\to A$ can be used to block certain rules when a nonterminal~$A$ is present by giving the blocked rule a lesser priority. This has already been utilized in Theorem~
\ref{thm:GCac-characterization-by-OCD-all-modes}.  
However, this does not seem possible for systems in $t$-mode, since a component can only be left when no more rules are applicable. This makes ordered grammar systems working in $t$-mode\LV{, in a way,} limited compared to ordered systems in other modes and leads us to\SV{:}\LV{ propose the following theorem.}

\begin{theorem}\label{thm:ordered-t}
    $\textsf{OCD}^{\,(-\emptyword)}(t)=\textsf{ORD}^{\,(-\emptyword)}$.
\end{theorem}
\begin{proof}
    Since an ordered grammar is equivalent to an ordered grammar system with one component, it is sufficient to show that an ordered grammar system with $n\geq 2$ components can be simulated by an adequate ordered grammar. \\
    Let $G=(N,\Sigma, (P_1,>_1), (P_2,>_2),\dots, (P_n,>_n), S)$ be an OCDGS with components $G_i=(N,\Sigma, P_i,>_i,S)$ for $i \in [n]$.\\
    Construct the ordered grammar $\Gamma=(N',\Sigma, P, >, S)$ with $$N'= \{S\} \cup \bigcup_{i\in[n],j \in [n],A\in N} \{A_i, A_{i\to j}\}\,.$$ 
    Define the homomorphisms $h_i:(N\cup\Sigma)^*\to (N'\cup\Sigma)^* $ by $a\mapsto a$ for $a\in \Sigma$ and $A\mapsto A_i$ for $A\in N$.
    Construct $P$ along with $>$ as follows:
    \begin{itemize} 
        \item Add rules $\{S \to S_i \mid i \in [n]\}$. \\
        Here, the first component to use is chosen. This is denoted by marking the starting symbol with~$i$ for component~$G_i$. These rules are the only rules with the starting symbol~$S$ and an unmarked nonterminal on the left-hand side.
        \item For each $P_i$, $i \in [n]$, add the (marked) rule $p_i = A_i \to h_i(w)$ for all $p= A\to w \in P_i$. 
        The order from $P_i$ is maintained with $p_i > q_i$ if $p >_i q$ for all $p,q \in P_i$.\\
        \item For each $P_i$ and for all $p\in P_i$, add the rules 
        $$\{X_k \to X_k 
        \mid k \in [n]-i 
        , X \in N \} > p_i > \{Y_i \to Y_{i\to j} \mid Y \in N, j \in [n]-i \} $$
         and for $j\in [n]-i$
        $$\{X_{l \to k} \to X_{l \to k} \mid X \in N,  l \in[n]-i, k \in [n]-j \}
        >\{Y_i \to Y_{i\to j} \mid Y \in N \}\,.$$
        As long as a nonterminal with marker $k\neq i$ is present in the sentential form, no rules from $P_i$ can be applied. When no more rules from $P_i$ are applicable, the next component~$G_j$ will be guessed and the marker of all nonterminals is updated to be the transition marker of the form $i\to j$. Note that all present nonterminals in the sentential form are marked. There cannot be two different transition markers in the sentential form. For an illustration of the order relation using an example \LV{of an OCDGS }with two components, \LV{refer to Figure}\SV{see Fig.}~\ref{fig:ocd2-to-ord}.

        \item For each $i,j \in [n]$ with $i\neq j$ add rules 
         \begin{align*}
        &\{X_\ell \to X_\ell,  X_{k \to \ell} \to X_{k \to \ell}  \mid X \in N, k \in [n],\ell \in [n]-j 
        \}\\&>
        \{ Y_{i\to j} \to Y_{j} \mid Y\in N\}\,. \end{align*}
         When all nonterminals have the same transition marker or  a marker for the next component, all nonterminals in the sentential form will be re-marked accordingly. To ensure that only the right transition is executed, the only rules here where symbols $X_{i\to j}$ are not blocked are $Y_{i\to j} \to Y_j$. 
         When all transition markers have been replaced, the derivation with rules from $P_j$ starts.   
    \end{itemize}
With these explanations at hand, a formal inductive proof for the correctness of the construction is straightforward but tedious and hence omitted.\qed 
\end{proof}
\begin{figure}[tb]
		\xymatrix@R=1em@C=2em{
			& \{X_{1\to2} \to X_{1\to2}\}_{X\in N} &  &  \\
			\{Y_{2\to1} \to Y_{1}\}_{Y\in N}  \ar[ru] \ar[dr] & \{X_{1} \to X_{1}\}_{X\in N}  & \ar[l] (P_2,>_2) & \ar[l] \ar[llu] \{Y_{2} \to Y_{2\to1}\}_{Y\in N}   \\
			\{Y_{1\to2} \to Y_{2}\}_{Y\in N}   \ar[ru] \ar[dr] & \{X_{2} \to X_{2}\}_{X\in N} & \ar[l] (P_1,>_1) &\ar[l]  \ar[lld] \{Y_{1} \to Y_{1\to2}\}_{Y\in N}  \\
			&\{X_{2\to1} \to X_{2\to1}\}_{X\in N}  \\
        }
    \caption{This graph shows the rules and order relations in a construction for two components. $p_1 \leftarrow p_2$ denotes $p_1 > p_2$. $(P_i,>_i)$ denotes all rules from the original component $G_i$ with its ordering. Note\LV{ that}\SV{:} there can be uncomparable rules within\LV{ that set}\SV{~$P_i$}.}\label{fig:ocd2-to-ord}    
	\end{figure}
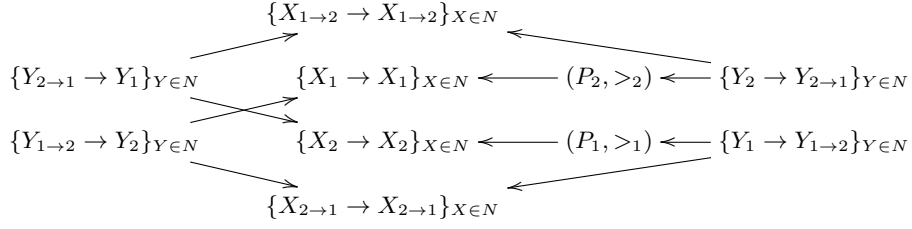

\section{Grammar Systems with Random Context Rules}
\label{sec:RCCDGS}

A \emph{random context cooperating distributed grammar system}, or RCCDGS for short,  is a construct $G=(N,\Sigma,P_1,\dots,P_n,S)$, where each component $G_i=(N,\Sigma,P_i,S)$ is a random-context grammar  (introduced in~\cite{Wal72}).
Each $P_i$ is a set of rules of the form $(A\to y,E,F)$, where $A\to y$ is a context-free rule, i.e., $A\in N$ and $y\in (N\cup\Sigma)^*$, $E,F\subseteq N$ are disjoint sets of permitting and forbidden symbols, respectively.
The derivation relation for each component is defined as follows.
For $u,v\in (N\cup\Sigma)^*$, $u\Rightarrow v$ if 
\begin{itemize}
    \item $\exists (A\to y\in R\, \exists x,z\in (N\cup\Sigma)^*: u=xAz, v=xyz$ and
    \item $(\forall A'\in P |u|_{A'}>0)\land(\forall A'\in F |u|_{A'}=0)$.
\end{itemize}
While the first condition is just the classical rewrite condition, the second one explains the meaning of random context. \SV{The derivation relation is defined as for OCDGS.}
\LV{
Let $\Rightarrow_i$ refer to the derivation relation of $G_i$. 
For $m=t$ (i.e., the $t$-mode of derivation), $u\Rightarrow^t_i v$ means $u\Rightarrow_i^*v$ and there is no $w$ such that $v\Rightarrow_i w$. The notations $\Rightarrow_i^{\leq k}$,  $\Rightarrow_i^{=k}$ and  $\Rightarrow_i^{\geq k}$ are referring to at most, exactly, and at least $k$ derivation steps of $G_i$ (as usual). We say that a derivation step $u\Rightarrow_i^m v$ is \emph{unproductive} if $u= v$; if $u\neq v$, we also say that we have applied $G_i$ \emph{productively}. In the simulations that we discuss in our paper, unproductive derivation steps can be safely ignored. 
If $G$ works in any of the modes~$m\in\Modes{1}$, then
the derivation relation $\Rightarrow_{G,m}$ is defined as $\bigcup_{i=1}^n\Rightarrow_i^m$.
Accordingly, $\cL(G,m)=\{w\in \Sigma^*\mid S\Rightarrow_{G,m}^*w\}$. }We can now define language families like $\textsf{RCCD}^{(-\emptyword)}(m)$, depending on the mode~$m$ of derivation. 
Also, we will consider grammar systems with all random context rules 
having empty permitting contexts, leading to $\textsf{fRCCD}^{(-\emptyword)}(m)$. When writing 
forbidden RC rules, it is sufficient to attach one set of nonterminals to the context-free rule. 
%
Masopust proved in \cite[Thm.~8]{Mas2009} that $\textsf{fRCCD}^{(-\emptyword)}(t)=\textsf{GC}_{\text{ac}}^{(-\emptyword)}$, so that in particular $\textsf{fRCCD}(t)=\textsf{RE}$ follows. No prior study has been done concerning the other modes to the best of our knowledge.

By just taking the well-known simulation between forbidden random context and ordered grammars, refer to~\cite{CreMay73}, one sees that these two mechanisms are indeed of the same power also when seen as grammar components. \LV{We therefore only provide a proof sketch.}\SV{Thus:}

\begin{theorem}\label{thm:forbiddenCD=orderedCD}
For all modes~$m\in\{=k,\leq k,\geq k\mid k\geq 1\}\linebreak[3]\cup\{*\}$, $$\textsf{fRCCD}^{\,(-\emptyword)}(m)=\textsf{OCD}^{\,(-\emptyword)}(m).$$
\end{theorem}

\begin{toappendix}
\begin{proof}[Sketch\SV{ for Theorem~\ref{thm:forbiddenCD=orderedCD}}]
Attach to a rule of an ordered CDGS the collection of left-hand sides of bigger rules (in that component) as its forbidden context. Conversely, let $X_f$ be a new symbol (failure symbol), 
and add in the component of the forbidden random context rule $(A\to w,F)$ bigger rules $X\to X_f$ for each $X\in F$, so that then $A\to w$ needs no forbidden context any longer. In both directions, it should be clear that a component in the simulating grammar can make $\ell$ steps if and only if the corresponding component in the simulated can perform $\ell$ steps.\qed
\end{proof}
\end{toappendix}


\begin{theorem}\label{thm:frccd1=frc}
For~$m\in\{\leq k\mid k\geq 1\}\linebreak[3]\cup\{*,=1,\geq1\}$, $\textsf{fRCCD}^{\,(-\emptyword)}(m)=\textsf{fRC}^{\,(-\emptyword)}$.
\end{theorem}

\begin{proof}[Sketch]
As any component can be executed also only one time, we can simply put all rules in a single component.\qed    
\end{proof}

In view of known hierarchy results concerning forbidden random context (or ordered) grammars (as stated, e.g., in \cite{Fer96a}), we can now immediately derive:

\begin{corollary}\label{cor:frccd1<frccd2}
    $\textsf{fRCCD}^{\,(-\emptyword)}(=1)\subsetneq \textsf{fRCCD}^{\,(-\emptyword)}(=2)$ and
    $\textsf{fRCCD}^{\,(-\emptyword)}(\geq 1)\subsetneq \textsf{fRCCD}^{\,(-\emptyword)}(\geq 2)$.
\end{corollary}

With the usual ``prolongation technique'' (formally applied below), one can show that  
$\textsf{fRCCD}^{\,(-\emptyword)}(\geq k)\subseteq \textsf{fRCCD}^{\,(-\emptyword)}(\geq k+1)$ for any $k$. Recall that such an inclusion chain for the $=k$-modes is unknown until today when considering classical CDGS. It is also unknown if these (possible) hierarchies collaps to a certain level.
In view of what we said, the following result is somehow surprising.

\begin{theorem} \label{thm:mode-hierarchy-collaps}
Let $k\geq 2$. Then,
$$\textsf{fRCCD}^{\,(-\emptyword)}(=2)=\textsf{fRCCD}^{\,(-\emptyword)}(=k)=\textsf{fRCCD}^{\,(-\emptyword)}(\geq k)\,.$$
\end{theorem}

\begin{proof}
First, we will show that $\textsf{fRCCD}^{\,(-\emptyword)}(\geq k)$ and $\textsf{fRCCD}^{\,(-\emptyword)}(=k)$ are included in $\textsf{fRCCD}^{(-\emptyword)}(=2)$. 
    Let $G_{\geq k}=(N,\Sigma, P_1, \dots, P_n, S)$ be an fRCCDGS in $\geq k$-mode with components~$G_i=(N, \Sigma, P_i, S)$ where $$P_i=\{(A_{i,1} \to w_{i,1}, F_{i,1}), \dots , (A_{i,m_i} \to w_{i,m_i}, F_{i,m_i})\}\,.$$ We will construct an equivalent RCCDGS $G_{=2}$ that works in $=2$-mode. In this construction, the essential objective is that at least $k$ rules from one component~$G_i$ have to be used before rules from a different component can be applied. With that, it is necessary to count how many derivation steps have been taken (at least until $k$ is reached) and to ensure that only rules from one component of~$G_{\geq k}$ are applied in this time. 
    Set 
    $$G_{=2}=(N_{=2}, \Sigma,P_{\text{init}}, P_{1,0}, P_{1,1}, \dots ,P_{1,k+1}, \dots , P_{n,0},P_{n,1},\dots ,P_{n,k+1},S')\,.$$
     Let $W_i\subseteq (N \cup \Sigma)^*$ be the set of all words that appear on the right-hand side of any rule in~$P_{i}$ of~$G_i$ and $X_i$ be the set of all nonterminals of the form $X_{\ell}^w$ for $\ell \in [0\dots k+1]$ and $w \in W_i$. Further, let $X_{\cup}$ denote $\bigcup_{i \in [n]} X_i$ and $Y_\cup$ denote $\{Y\} \cup \bigcup_{i \in [n]} \{Y_i\}$; finally, set $N_{=2}=N\cup X_{\cup} \cup Y_\cup$. \\
    For each $G_i$ of $G_{\geq k}$, components~$G_{i,\ell}$, for $\ell \in [k+1]$, will be constructed. Since the set of permitting symbols is always empty, it will be omitted. 
    \\
    The very first derivation step has to be taken by component~$G_{\text{init}}$ which contains the only rule with the new starting symbol~$S'$ on the left-hand side of a rule: $P_{\text{init}}=\{(S\to S,\emptyset), (S'\to YS, \emptyset)\}$. 
    Afterwards, the RCCDGS~$G_{=2}$ can take a step in accordance with the rules of~$G_{\geq k}$. The nonterminals of the set~$X_i$ are used to count the number of derivations up to $k$ and are therefore necessary as a conector to the next component. They further store the word $w_{i,j}$ from the rule $(A_{i,j}\to w_{i,j},F_{i,j})$ of the ``original'' component~$G_i$.
     $$P_{i,0}=\{(A_{i,j}\to X_{1}^{w_{i,j}}, F_{i,j}\cup X_\cup \cup Y_\cup-Y) \mid j \in [m_i]\} \, \cup \{(Y \to Y_i,\emptyset ) \}$$
    The nonterminal $Y_i$ ensures that only rules from the ``original'' component~$G_i$ can be used. The corresponding nonterminals signifying other components (e.g., $Y_j$ with $j\neq i$) are forbidden in all rules of the current component. \\
    For all $1\leq\ell<k$: 
 
    $$\begin{array}{rl}
        P_{i,\ell}=&\{(X_{\ell}^{w_{i,j}}\to w_{i,j}, Y_\cup-Y_i) \mid j \in [m_i] \}\, \cup \\
    &\{(A_{i,j}\to X_{\ell+1}^{w_{i,j}}, F_{i,j}\cup X_\cup \cup  Y_\cup-Y_i) \mid j \in [m_i]\}
    \end{array}$$
    
    The word $w_{i,j}$ is derived and then, the next derivation step according to~$P_i$ is simulated. The counter is updated and 
    the word~$w_{i,j}$ is again stored in the nonterminal $X_{\ell+1}^{w_{i,j}}$. This is done to avoid $\emptyword$-rules. The counter reflects how many derivation steps have already been taken. 
    
    $$\begin{array}{rl}
    P_{i,k}=& \{(X_{k}^{w_{i,j}}\to w_{i,j},  Y_\cup-Y_i) \mid j \in [m_i] \}\,\cup \\
    &\{(A_{i,j}\to X_{k}^{w_{i,j}}, F_{i,j}\cup X_\cup \cup Y_\cup-Y_i) \mid j \in [m_i]\}
    \end{array}$$
    
    Since the minimum of $k$ steps has been reached, the counter does not need to be incremented anymore. The component $G_{i,k}$ can be reentered an arbitrary amount of steps to achieve $\geq k$. To leave the ``original'' component~$G_i$, enter component~$G_{i,k+1}$ .
    
    $$\begin{array}{rl}
    P_{i,k+1}=&\{(X_{k}^{w_{i,j}}\to w_{i,j},  Y_\cup-Y_i) \mid j \in [m_i] \}\, \cup \\
     & \{(Y_i \to Y,  Y_\cup-Y_i \cup X_i), (Y_i \to \lambda, N-Y_i) \}
    \end{array}$$
    
    In  component~$G_{i,k+1}$, the last stored word is written and the signifier~$Y_i$ is reset to $Y$, denoting that a new component can be chosen. If $Y_i$ is the last nonterminal left, it can be deleted to end the computation.
    
    By omitting the component~$G_{i,k}$, this construction also works to simulate a fRCCDGS working in $=k$-mode.\\

    Next, we will show that $\textsf{fRCCD}^{(-\emptyword)}(=2)$ is included in $\textsf{fRCCD}^{(-\emptyword)}(=3)$ and $\textsf{fRCCD}^{(-\emptyword)}(\geq3)$.
     Let $G_{=2}=(N,\Sigma, P_1, \dots, P_n, S)$ be an fRCCDGS in $=2$-mode with components~$G_i=(N, \Sigma, P_i, S)$ where $$P_i=\{(A_{i,1} \to w_{i,1},  F_{i,1}), \dots , (A_{i,m_i} \to w_{i,m_i},  F_{i,m_i})\}\,.$$

     \noindent
    Construct the fRCCDGS $G_{=3}=(N_{=3},\Sigma, P_{3,1}, \dots, P_{3,n}, P_{3,\text{reset}}, S)$ in $=3$-mode with $N_{=3}=N\cup N' \cup N'' \cup N''' \cup X_\cup$ where $X_\cup=\bigcup_{i \in n, (A \to w,F) \in P_i}X_w$. For each component $G_i$ of $G_{=2}$, \LV{will }construct a\SV{n} \LV{equivalent }component~$G_{3,i}$ where three derivation steps are taken. Let $G_i$ be defined as above, then:
    $$\begin{array}{rl}
    P_{3,i}= &\{(A_j \to X_{w_j}, F_j \cup N' \cup N'' \cup N''' \cup X_\cup) \mid j \in [m_i]\}\, \cup\\
                &\{ (X_{w_j} \to w_j', F_j \cup N' \cup N'' \cup N'''\cup X_\cup-X_{w_j}) \mid j \in [m_i]
               \}\\
               & \{(A_j \to w_j'', F_j \cup N'' \cup N''' \cup X_\cup) \mid j \in [m_i] \}\, \cup\\
               
               & \{(A_j' \to w_j''', F_j \cup N'' \cup N''' \cup X_\cup) \mid j \in [m_i]\}\, \cup\\
    \end{array}$$
    \noindent
    where $w_j' = a'v$ for $w_j = av$ with $a \in N\cup T \cup \{\emptyword\}, a' \in N', v \in (N \cup T)^*$; $w_j''$ and $w_j'''$ are defined analogously with the first symbol belonging to $N''$ and $N'''$, respectively. The markers on the first symbol of the derived word denote that the derivation took place. The forbidding property and the prolongation via $X_{w_j}$ ensure that only two rules from $G_i$ are applied. In case the first symbol of~$w_j'$ is a nonterminal that can be derived in the second step, the set of extra rules with $A_j'$ on the left-hand side is added. With that, it is possible to denote that two derivation steps took place even though no symbol from $N'$ is in the sentential form. \\
   Before the next component can be entered, the sentential form needs to be reset to only contain nonterminals from~$N$. This will be done with\LV{ the component} $G_{3,\text{reset}}$:
    $$\begin{array}{rl}
    P_{3,\text{reset}}= 
        & \{(A'\to A, X_\cup) \mid A' \in N'\}  \cup \{(A''\to A', X_\cup ) \mid A'' \in N''\} \, \cup \\
        & \{(A'''\to A'', X_\cup ) \mid A''' \in N'''\}  
    \end{array}$$
     When a component~$P_{3,i}$ is left, the sentential form either contains exactly one symbol from $A'''$ or one symbol from $N'$ and $N''$ each. It can easily be observed that three steps need to be taken in the component~$P_{3,\text{reset}}$ to transform all present symbols from $N'\cup N''\cup N'''$ into symbols from~$N$. 

     The construction can be adapted for $\geq 3$-mode by adding suitable unproductive rules. Further, it is easily seen that the adaption for any $k\geq 3$ is possible.\LV{ Therefore, the proposed equality holds.} \qed   
\end{proof}


\begin{toappendix}
Concerning the number of components as a complexity measure, Masopust~\cite[Thm.~1]{Mas2009} has shown that for $\textsf{fRCCD}^{(-\emptyword)}(t)$, two components suffice to generate all these languages. Theorem~\ref{thm:ordered-t} tells us that for ordered CDGS working in the $t$-mode, even one component is sufficient. If we denote (as usual with CDGS) the bound on the number of components by a subscript, we obtain:

\begin{corollary}
$\textsf{OCD}_1^{\,(-\emptyword)}(t)= \textsf{OCD}^{\,(-\emptyword)}(t)=\textsf{fRCCD}_1^{\,(-\emptyword)}(t)\subsetneq \textsf{fRCCD}^{\,(-\emptyword)}(t).$
\end{corollary}
\end{toappendix}

\section{CD Grammar Systems With Internal Control}
\label{sec:CDGS-rc}

This is exactly the title of Section~4.2 of \cite{Csuetal94all}. We will now present a variant of Def.~4.4 from that book, tailored towards our purposes. The technical difference is that in the monograph, components might have symbols where the role of nonterminals and terminals is potentially interchanged. However, it is easy to check that the construction of  \cite[Lemma 3.6]{Csuetal94all} that shows that this more general definition is not giving more descriptive capacity also holds for \SV{our}\LV{the} variants\LV{ that we discuss in this section}.\LV{ Moreover, this justifies why we took this more restrictive version of the definition of CDGS as a basis, as has been done also in the literature from the mid 90s onwards. In particular, the forbidden context variants discussed above have always been defined as we did it in this paper.}

A CDGS with \emph{random context entry conditions} is given by a construct $G=(N,\Sigma,(P_1,E_1,F_1),\dots,(P_n,E_n,F_n),S)$, where any component $G_i=(N,\Sigma,P_i,S)$ is a context-free grammar and $E_i,F_i\subseteq N$ specify the \emph{entry conditions} of \LV{this component}\SV{$G_i$}.
Let $\Rightarrow_i$ \LV{refer to}\SV{be} the derivation relation of $G_i$\SV{ and $m\in\Modes{1}$}.\LV{ 
We again consider all classical modes.\footnote{They correspond to the exit conditions in Def. 4.4. of \cite{Csuetal94all}.}}
Now, $u\Rightarrow_i^{m,\text{rc}}$ \SV{iff}\LV{should mean that} 
\begin{itemize}
    \item  $u\Rightarrow_i^{m}$ (derivation condition), as well as
    \item 
$|u|_{A}>0$ for all $A\in E_i$ and $|u|_{A}=0$ for all $A\in F_i$ (random context entry conditions).
\end{itemize}
Here, $m\in\Modes{1}$.
As before, we define the derivation relation $\Rightarrow_{G,m,\text{rc}}$ as $\bigcup_{i=1}^n\Rightarrow_i^{m,\text{rc}}$. Accordingly, $\cL(G,m,\text{rc})=\{w\in \Sigma^*\mid S\Rightarrow_{G,m,\text{rc}}^*w\}$. We can now define language families like $\textsf{CD}^{(-\emptyword)}(m,\text{rc})$, depending on the mode~$m$ of derivation. 
Also, we will consider 
grammar systems with all random context conditions 
having empty permitting contexts, leading to $\textsf{CD}^{(-\emptyword)}(m,\text{frc})$. When writing 
forbidden RC conditions, it is sufficient to attach one set of nonterminals to the components. No prior study has been done concerning the restriction to permitting or forbidden context conditions to the best of our knowledge.

From \cite{Csuetal94all}, we get the following result, rephrased in our terminology:
\begin{theorem}\label{thm:CsuDasKelPau-book}
Let $k\geq 1$ be arbitrary.
   $\textsf{CD}^{\,(-\emptyword)}(= k,\text{rc})= \textsf{GC}_{ac}^{\,(-\emptyword)}$.
\end{theorem}

\begin{theorem}\label{thm:frc-CD}
  $ \textsf{CD}^{\,(-\emptyword)}(t,\text{frc})=\textsf{CD}^{\,(-\emptyword)}(t)$.  
\end{theorem}

\begin{proof}
The inclusion $\supseteq$ is trivial. For the other direction, consider a CDGS  with forbidden random context entry conditions, given by the construct $G=(N,\Sigma,(P_1,F_1),\dots,(P_n,F_n),S)$.
Let $N'=\{A'\mid A\in N\}$ and $\tilde{N}=N\cup N'$. Let $h:(N\cup \Sigma)^*\to (N'\cup \Sigma)^*$ be the homomorphism defined by $A\mapsto A'$ for $A\in N$ and $a\mapsto a$ for $a\in \Sigma$.
Define the CDGS $\tilde{G}=(\tilde{N},\Sigma,\tilde{P}_0,\tilde{P}_1,\dots, \tilde{P}_n,S)$ with 
$\tilde{P}_0=\{A'\to A\mid A\in N\}$ and, for $i\in [n]$, $\tilde{P}_i=\{A\to h(w)\mid A\to w\in P_i\land A\notin F_i\}\cup\{h(A)\to h(w)\mid A\to w\in P_i\}\cup\{X\to X\mid X\in F_i\}$.
A derivation step $u\Rightarrow_{P_i}v$ of $G$ is faithfully simulated by $u\Rightarrow_{\tilde{P}_i}u'\Rightarrow_{\tilde{P}_0} v $. In particular, a rule $A\to w\in P_i$ with $A\in F_i$ cannot applied in the very first step of component $P_i$, but it might be applied in a later step, which is possible within $\tilde{G}$ as well.
\qed 
\end{proof}

\begin{toappendix}
\begin{remark}
As $\textsf{CD}^{\,(-\emptyword)}(t)$ are the E(P)T0L languages which are strictly contained in the ordered languages (see~\cite{Pen75}), this shows that this ``internal control'' mechanism introduced in the spirit of \cite{Csuetal94all} is strictly weaker than the attachment of ordered or also forbidden random context to the individual rules (also in the $t$-mode) as studied in the previous sections.
Interestingly, looking at the proof of Lemma 4.10, using in turn Lemma 4.9, in \cite{Csuetal94all}, one can notice that this proof also works in the sense that it shows that E(P)T0L systems with forbidden random contexts can be simulated by $\textsf{CD}^{\,(-\emptyword)}(t,\text{frc})$. Hence, Theorem~\ref{thm:frc-CD} can be also used to prove that forbidden random context does not increase the power of ET0L systems, a result that was also obtained by Meduna and Svec~\cite{MedSve2003}\LV{ in different ways}.
\end{remark}    
\end{toappendix}

\begin{lemma}\label{lem:CDfrc-in-fRCCD}
    For any 
    $m \in \Modes{1}$, we find
    $\textsf{CD}^{\,(-\emptyword)}(m,\text{frc}) \subseteq \textsf{fRCCD}^{\,(-\lambda)}(m)$. 
\end{lemma}
\begin{proof}[Sketch]
    Let $G=(N,\Sigma,(P_1,F_1),\dots,(P_n,F_n),S)$ be an frcCDGS. In the corresponding fRCCDGS $\tilde{G}=(\tilde{N}, \Sigma, \tilde{P_0}, \tilde{P_1}, P_{F_1}, \dots, \tilde{P_n}, P_{F_n}, \tilde{P}_{\text{end}}, \tilde{S})$ with $\tilde{N}=N \cup \{Y\} \cup \{Y_{F_i}\}_{i \in [n]}$, the first step will introduce a symbol to denote which forbidden context has been checked: $\tilde{P_0}=\{\tilde{S} \to YS \}$.
    The component~$P_{F_i}=\{(Y \to Y_{F_i}, F_i)\}\cup \{(Y_{F_j} \to Y_{F_i}, F_i) \mid j \in [n]\}$ checks that the forbidden context of $P_i$ is fulfilled. Then, the derivation according to component $P_i$ can begin with $\tilde{P_i}=\{(A \to w, \{Y\} \cup \{Y_{F_j}\}_{j \in [n]-i}) \mid A \to w \in P_i\}$. At the end of the derivation, i.e., when no other nonterminals are left, delete the symbol $Y_{F_i}$: $\tilde{P}_{\text{end}}=\{(X \to \lambda, N) \mid X\in \{Y\} \cup \{Y_{F_i}\}_{i \in [n]} \}$.
    In $t$-mode, simply omit the rules of the form $Y_{F_i}\to Y_{F_i}$ from $P_{F_i}$ to avoid loops. This will not have any further effects since a component cannot be entered twice in a row in $t$-mode anyways, making the rule unnecessary.
    \qed
\end{proof}

\noindent
The converse direction is more demanding, as we will see now.

\begin{lemma}\label{lem:fRCCDinCDfrc=2}
    $\textsf{fRCCD}^{\,(-\lambda)}(=2) \subseteq \textsf{CD}^{\,(-\lambda)}(=2,\text{frc})$. 
\end{lemma}
\SV{\begin{proof}[Sketch]
    Let $G=(N,\Sigma, P_1, P_2, \dots, P_n, S)$ be an fRCCDGS with components $G_i=(N,\Sigma, P_i, S)$.
    Construct frcCDGS $\tilde{G}$ with the start symbol~$S$, terminal alphabet~$\Sigma$ and nonterminal alphabet $\tilde{N}=N \cup \{\#\} \cup X_{\cup}$ with $X_{\cup} = \bigcup_{i\in[n], p\in P_i} \{X_p\} \cup \bigcup_{i\in [n]} \{X_{A\to \xi A'\eta \to \xi w'\eta,F,F'} \mid(A\to \xi A'\eta,F), (A'\to w',F') \in P_i\}$.
    For each component $P_i$, $i\in [n]$, we introduce several components, each containing at most two rules, into $\tilde{G}$, as described next.
    \begin{enumerate}
        \item Add component $P_{p}=\{A \to w\}$ with forbidden set $F_{p}=F\cup X_{\cup} \cup \{\#\}$ for all $p=(A\to w, F)$ where $|w|_F=0$, i.e., \LV{all rules in $P_i$ that}\SV{$p$} can be executed twice in a row.
        
        \item For all $p=(A\to w,F),p'=(A' \to w', F') \in P_i$ (with $p\neq p'$) that can be executed one after the other, i.e., 
        \begin{itemize}
            \item if $\{A,A'\} \cap (F\cup F') =\emptyset$ and ($|w|_{F'}=0$ or $|w'|_F=0$), or
            \item if $A\in F',A' \notin F$ and $|w|_{F'}=0$, or
            \item if $A\notin F',A' \in F$ and $|w'|_F=0$,
        \end{itemize}
        then add components \begin{itemize}
            \item 
        $P_{X,p}=\{A \to \#, \# \to X_{p}\}$ with $F_{X,p}=F\cup
        \{\#,X_{p}\} $ and 
        \item  $P_{X,p'}=\{A' \to \#, \# \to X_{p'}\}$ with $F_{X,p'}=F'\cup \{\#,X_{p'}\} $.
        \end{itemize} 
        Now, the nonterminals that will be derived have already been chosen and the order in which the rules are executed is irrelevant. Therefore, add component $P_{p, p'}=\{ X_{p} \to w, X_{p'} \to w'\}$ with $F_{p, p'}=F\cup F' \cup (X_{\cup} \setminus \{X_{p}, X_{p'}\}) \cup \{\#\}$.

        \item Add the component $P_{A\to \xi A'\eta \to \xi w'\eta,F,F'}=\{X_{p} \to \#, \# \to \xi A'\eta \}$ with the forbidden context $F_{A\to \xi A'\eta \to \xi w'\eta,F,F'}= F \cup F' \cup (X_{\cup} -X_{p}) \cup \{\#\} $ for all rules $p=(A\to w, F), p'=(A'\to w', F') \in P_i$ (with $(p\neq p'$) where $w=\xi A'\eta$ with $\xi,\eta\in (N\cup \Sigma)^*$ and $|\xi\eta|_{F'}=0$. The marking of exactly one $A$ (as in the afore-mentioned case) is necessary in case $A\in F$.  \qed
    \end{enumerate}
\end{proof}}

\begin{toappendix}
\SV{\medskip
\noindent
    \textbf{Proof of Lemma~\ref{lem:fRCCDinCDfrc=2}}}
    
    \begin{proof}
    To prove the correctness of this construction, consider first a derivation step $x\Rightarrow_i^{=2}y$ of $G$.
    This means that there are two rules $p=(A\to w,F),p'=(A' \to w', F') \in P_i$ such that $x\Rightarrow_p x'\Rightarrow_{p'}y$ or $x\Rightarrow_{p'} y'\Rightarrow_py$. W.l.o.g., consider the first case. As $x\Rightarrow_p x'$, $|x|_F=0$. Assume $x=uAv$ and $x'=uwv$. As $x'\Rightarrow_{p'}y$, furthermore, $|uwv|_{F'}=0$. In order to derive $y$, one occurrence of $A'$ within $x'$ is replaced by $w'$, i.e., $x'=u'A'v'$ and $y=u'w'v'$.
    
    We consider several situations that can happen:
    Firstly, we could find $p=p'$. Then, we see two subcases where the second replacement of $A$ takes place.
        \begin{itemize}
            \item If $|w|_A>0$, then we might see $w=w_1Aw_2$, so that $y=uw_1ww_2v$. This is only possible if $|uw_1Aw_2v|_F=0$. In the simulating grammar system, we add the component $P_p$ only if $|w|_F=0$. Moreover, the forbidden context set  $F_p$ contains $F$, so that $|x|_F=0$ is checked. Taken together, $|uAv|_F=0$ and $|w_1Aw_2|_F=0$ imply $|uw_1Aw_2v|_F=0$. Hence, such a derivation of $P_i$ is possible if and only component  $P_{p}$ of $\tilde{G}$ is applicable in the described way.
            \item If $|x|_A>1$, then on $x=x_0Ax_1Ax_2$, we could replace the two mentioned occurrences of $A$ in any order when applying $P_i$ of~$G$; this applicability also means that $|x|_F=0$ and $|w|_F=0$ for the second application. By having $F$ in its forbidden context set, $P_{p}$ checks that $|x|_F=0$; furthermore, this component is only added to $\tilde{G}$ if  $|w|_F=0$. Hence, $x\Rightarrow_i^{=2}y$ of $G$ (in the described case) is possible to simulate by $x\Rightarrow_{{p}}^{=2}y$ of $\tilde{G}$. (We continue using the subscripts of the components as subscripts of $\Rightarrow$ also if these are not simply numbers.)
        \end{itemize}
        Secondly, we might have $p\neq p'$. Again, we consider two similar subcases where the replacement of $A'$ happens.
        \begin{itemize}
            \item If $|w|_{A'}>0$, then we might see $w=w_1A'w_2$, so that $y=uw_1w'w_2v$. This is only possible if $|uw_1Aw_2v|_{F'}=0$. In the simulating grammar system, we have the components $P_{X,p}$ and $P_{A\to w_1A'w_2 \to w_1w'w_2,F,F'}$. We could derive in $\tilde{G}$ with $$x=uAv\Rightarrow_{X,p}^{=2}=uX_{p}v\Rightarrow_{A\to w_1A'w_2 \to w_1w'w_2,F,F'}^{=2}uw_1w'w_2v=y$$
            if $|x|_{F_{X,p}}=0$ (so that $|x|_{F}=0$) and $|uX_{p}v|_{F_{A\to w_1A'w_2 \to w_1w'w_2,F,F'}}=0$ (implying  $|uv|_{F'}=0$), as well as $|w_1w_2|_{F'}=0$ by construction. Hence, $x\Rightarrow_i^{=2}y$ of $G$ (in the described case) is possible  to simulate in $\tilde{G}$.
            \item Also, we might see that $x=x_0Ax_1A'x_2$ (the case when $A'$ occurs left of $A$ is similar) and $y=x_0wx_1w'x_2$. In~$G$, this is only possible if $|x|_F=0$ and if $|x_0wx_1A'x_2|_{F'}=0$. This can be simulated by $\tilde{G}$ as follows. First, the component $P_{X,A\to w, F}$ is applied, leading to $x\Rightarrow_{X,p}^{=2}x_0X_{p}x_1A'x_2$; this is possible only if $|x|_F=0$. Secondly, apply $P_{X,p'}$, so that $$x_0X_{p'}x_1A'x_2\Rightarrow_{X,p'}^{=2}x_0X_{p}x_1X_{p'}x_2\,;$$ also passing the test that $|x_0X_{p}x_1A'x_2|_{F'}=0$ as $|x_0wx_1A'x_2|_{F'}=0$. Finally, the component $P_{p,p'}$ is applied, leading to $$x_0X_{p}x_1X_{p'}x_2\Rightarrow_{p,p'}^{=2}x_0wx_1w'x_2=y\,,$$ passing the test $|x_0X_{p}x_1X_{p'}x_2|_{F_{p,p'}}=0$ as also $|x_0X_{p}x_1X_{p'}x_2|_{F\cup F'}=0$.
        \end{itemize}
    Conversely, consider some sentential form~$x$ that is derivable in $\tilde{G}$. As $\#$ is in the forbidden context of all components, $x\Rightarrow_{\tilde{G}}y$ is not possible for any $y$ if $|x|_\#>0$.
    This means that any component that makes use of~$\#$ must apply the two rules in the given order. We will use this simple observation frequently.

    Now, we discuss the different components one-by-one. 
    
    Firstly, assume that  $x\Rightarrow_{p}^{=2}y$. As $F_{p}$ is constructed, $x\in ((N\setminus F)\cup\Sigma)^*$. Moreover, this component is only introduced if $|w|_F=0$. Also, there is some $x'$ such that $x\Rightarrow x'\Rightarrow y$ by applying the context-free rule ${p}$. As $|x|_F=0$ and $|w|_F=0$, we can also consider using the forbidden random context rule $p=(A\to w,F)$ twice to first get $x'$ from $x$ and then $y$ from $x'$. By construction, $P_{p}$ and $F_{p}$ have been introduced into $\tilde{G}$ because the rule $p$ was contained in some component $P_i$. Therefore, $x\Rightarrow_i^{=2}y$ is possible in~$G$.

    Next, assume that $x\Rightarrow_{X,p}^{=2}y$. As $F_{X,p}=F\cup\linebreak[3] \{\#,X_{p}\} $, in particular $|x|_{X_{p}}=0$, so that there can never be more than one occurrence of each of the nonterminals  $X_{p}$ in any sentential form  derivable in $\tilde{G}$, assuming that we finally derive some terminal word from it, as by our previous observation, we can assume that only $|y|_\#=0$ is interesting. However, at this point one might think that such sentential forms may contain an arbitrarily large number of different nonterminals of the form $X_{p}$, corresponding to different rules $p=(A\to w,F)$ of~$G$. However, as (most of) $X_\cup$ is in the forbidden context of all other components, only those sentential forms can possibly derive terminal strings that contain at most two occurrences of nonterminals from $X_\cup$, and these two nonterminals must ``match''.

    Thirdly, assume that $x\Rightarrow_{p,p'}^{=2}y$ with $p=(A\to w,F), p'=(A'\to w', F')$. As $P_{p,p'}=\{ X_{p} \to w, X_{p'} \to w'\}$, $x$ must have at least one occurrence of $X_{p}$ and  at least one occurrence of $X_{p'}$. By our previous reasoning, $|x|_{X_{p}}=|x|_{X_{p'}}=1$. Assume that $x=x_0X_{p}x_1X_{p'}x_2$. Then, we must see $y=x_0wx_1w'x_2$. Recall that $X_{p}$ has been previously introduced by applying $P_{X,p}$ and similarly, $X_{p'}$ has been previously introduced by applying $P_{X,p'}$. By construction, there is a component $P_i$ of $G$ that contains both the forbidden random context rules $p$ and $p'$. In this sense, the two nonterminals $X_{p}$ and $X_{p'}$ have to ``match''. Moreover, by our observations above, there must exist sentential forms $u,v$ such that $|u|_{X_\cup}=0$ and (w.l.o.g., concerning the order of applications, as the discussion of the other ordering is symmetric) $$u\Rightarrow_{X,p}^{=2}v\Rightarrow_{X,p'}^{=2}x=x_0X_{p}x_1X_{p'}x_2\Rightarrow_{p, p'}^{=2}y\,.$$
    Hence, $u=x_0Ax_1A'x_2$. As we could apply $P_{X,p}$ on $u$, $|u|_F=0$, so that we could derive $u'=x_0wx_1A'x_2$ by applying the forbidden random context rule $p$. Furthermore, we know that $|v|_{F'}=0$, i.e., $|x_0X_{p}x_1A'x_2|_{F'}=0$. If $|w|_{F'}=0$, this would mean that it would be possible to apply  the forbidden random context rule $p'$ on $u'$ in order to obtain~$y$. Hence, $u\Rightarrow_i^{=2}y$ is possible as a derivation step in~$G$.
The given reasoning shows how a possible derivation in~$\tilde{G}$ corresponds to a derivation in~$G$ for the first two cases of the following list taken from the definition of~$\tilde{G}$ concerning the question when the simulating components are actually introduced:
 \begin{itemize}
            \item if $\{A,A'\} \cap (F\cup F') =\emptyset$ and ($|w|_{F'}=0$ or $|w'|_F=0$), or
            \item if $A\in F',A' \notin F$ and $|w|_{F'}=0$, or
            \item if $A\notin F',A' \in F$ and $|w'|_F=0$.
        \end{itemize}
The last case covers the case when (on the string~$u$), first $P_{X,p'}$ is applied and then $P_{X,p}$, because then we see $$u\Rightarrow_{X,p'}^{=2}v'=x_0Ax_1w'x_2\Rightarrow_{X,p'}^{=2}x=x_0X_{p}x_1X_{p'}x_2\,,$$
checking $|u|_{F'}=0$ and $|x_0Ax_1w'x_2|_F=0$, so that $|w'|_F=0$ but possibly $A' \in F$.
Hence, $u\Rightarrow_i^{=2}y$ is possible as a derivation step in~$G$ in all possible situations that could have led to $x$ in~$\tilde{G}$.

    Lastly, assume that $x\Rightarrow_{A\to \xi A'\eta \to \xi w'\eta,F,F'}^{=2}y$. As observed earlier, we can assume that $x=x_1 X_{A\to w, F}x_2$ and $y=x_1 \xi w'\eta x_2$. Also, we know that $x$ was obtained (within the derivation of~$\tilde{G}$) from some $v$ with $v=x_1 Ax_2$ by applying $P_{X,A\to w, F}$. By construction, there is a component $P_i$ that contains the rules 
    $(A\to w,F)$ with $w=\xi A'\eta$ and $(A'\to w',F')$. By having $v\Rightarrow_{X,A\to w, F}^{=2}x$, $|v|_F=0$ was checked, corresponding to deriving $v'=x_1wx_2$ from $v$ by applying  $(A\to w,F)$. Upon applying    
    $P_{A\to \xi A'\eta \to \xi w'\eta}$, we check $|x_1 X_{A\to w, F}x_2|_{F'}=0$. Moreover, we have introduced this component only if $|\eta\xi|_{F'}=0$. Hence, indeed it is possible to apply  $(A'\to w',F')$ on $v'=x_1\eta A'\xi x_2$, deriving $y$. Hence, $v\Rightarrow_i^{=2}y$ holds in~$G$.

The previously discussed four cases can be considered as the induction step of the following claim that then follows by an easy induction argument.
\emph{Let $u\in (N\cup\Sigma)^*$. If $u$ is derivable from $S$ in $\tilde{G}$, then $u$ is derivable from $S$ in~$G$.}
\qed
\end{proof}
\end{toappendix}

\noindent
The previous two lemmas are the cornerstones to show:

\begin{theorem}\label{thm:frc=k}
    $\textsf{fRCCD}^{\,(-\lambda)}(=2) = \textsf{CD}^{\,(-\lambda)}(=k,\text{frc})$ for all $k\geq 2$.
\end{theorem}

This theorem generalizes Theorem~\ref{thm:CsuDasKelPau-book} considerably, as there, also permitting context conditions have been used. This theorem shows that this is not necessary.  

\begin{proof}[Sketch]
From Lemma~\ref{lem:fRCCDinCDfrc=2} in particular, we know that the case $k=2$ is settled with Lemma~\ref{lem:CDfrc-in-fRCCD}. For the bigger values of $k$, we only need to generalize the construction of Lemma~\ref{lem:fRCCDinCDfrc=2} because of Theorem~\ref{thm:frc=k}.
First, notice that we introduced the very first case (concerning a double application of $A\to w$) into the proof of Lemma~\ref{lem:fRCCDinCDfrc=2} only in order to reduce the number of nonterminals. Alternatively, we could have used two $X$-type variables, say, $X_{A\to w, F}$ and $X_{A\to w, F}'$, that can then be handled by the second and third cases as described in that proof. By doing so, we end up with using more nonterminals, but we gain that all components contain exactly two rules. If  this component is of the form $\{Y\to\#,\#\to y\}$, then it is trivially to ``prolong'' the process by adding at most $k-2$ new nonterminals $\#_i$ so that one can count up to $k-2$ waiting steps. Otherwise, we see two different nonterminals $Y,Z$ with their corresponding rules in $\{Y\to y, Z\to z\}$, and again, we can use our `counting nonterminals'  $\#_i$ to prolong one of the two expected derivations. This way, it is clear that $\textsf{fRCCD}^{\,(-\lambda)}(=2) \subseteq \textsf{CD}^{\,(-\lambda)}(=k,\text{frc})$. \qed
\end{proof}

\SV{Notice that this is a considerable strengthening of Theorem~\ref{thm:CsuDasKelPau-book} where the full power of random context conditions have been used to prove such a result.}
\begin{toappendix}
\begin{remark}    
Notice that this is a considerable strengthening of Theorem~\ref{thm:CsuDasKelPau-book} where the full power of random context conditions have been used to prove such a result (namely, compare also with Theorems~\ref{thm:GCac-characterization-by-OCD-all-modes} and~\ref{thm:forbiddenCD=orderedCD}). 
However, we are not sure at all about the power of $\textsf{CD}^{\,(-\lambda)}(\geq k,\text{frc})$ if $k\geq 2$. These systems might be strictly weaker than $\textsf{fRCCD}^{\,(-\lambda)}(\geq k)$.
\end{remark}
\end{toappendix}

\section{CD Grammar Systems With Priorities}
\label{sec:PCDGS}

We can observe further connections to the literature. In \cite{MitPauRoz94}, so-called \emph{grammar systems with priorities} have been introduced; 
they are defined as follows.
A cooperating distributed grammar system with priorities 
is given by a construct $G=(N,\Sigma,P_1,\dots,P_n,>,S)$, where $(N,\Sigma,P_1,\dots,P_n,S)$ is a CDGS is a strict ordering on $\{P_1,\dots,P_n\}$.
Let $\Rightarrow_i$ refer to the derivation relation of $G_i=(N,\Sigma,P_i)$. 
We again consider all classical modes~$m$. 
Now for $x,y\in (N\cup \Sigma)^*$, $x\Rightarrow_i^{m,>}y$ if and only if $x\Rightarrow_i^{m}y$ and for no $j\in [n]$ with $P_j>P_i$, there exists a $z\in (N\cup \Sigma)^*$ such that $x\Rightarrow_j^{m,>}z$. Let $\Rightarrow_{m,>}=\bigcup_{i\in[n]}\Rightarrow_i^{m,>}$.
This way, we can finally define the generated language $\cL_m(G)=\{y\in (N\cup \Sigma)^*\mid S\Rightarrow_{m,>}^*y\}$. We arrive at language classes like $\textsf{PCD}^{(-\emptyword)}(m)$.

\begin{remark}\label{rem:PCD-survey}
From~\cite{MitPauRoz94}, we know the following inclusion relations (in our terminology, employing some further well-known results from regulated rewriting; the numbers refer to the theorems of that paper):
\begin{enumerate}\setcounter{enumi}{1}
\item $\textsf{CD}^{\,(-\emptyword)}(t)=\textsf{PCD}^{\,(-\emptyword)}(t)$.
\item  $\textsf{ORD}^{\,(-\emptyword)}\subseteq \textsf{PCD}^{\,(-\emptyword)}(m)$ for $m\in\{=k,\leq k \mid k\geq 1\}$.
    \item  $\textsf{PCD}^{\,(-\emptyword)}(m)\subseteq \textsf{GC}^{\,(-\emptyword)}_{\text{ac}}$ for  $m\in\{=k,\leq k,\geq k\mid k\geq 1\}\linebreak[3]\cup\{t\}$.
\end{enumerate}
\end{remark}

Now, we compare PCDGS with CDGS with frc entry conditions. On an intuitive level, this correspondence looks very reasonable, as both variants of CDGS apply either the ``ordered paradigm'' or the ``forbidden random context'' paradigm to components and their applicability, instead of associating them to the rules of any component as we have studied in the preceding sections. However, there are some technicalities that need to be taken into consideration.

\begin{lemma}\label{lem:cd<pcd}
For all modes $m\in\Modes{1}$,
 $\textsf{CD}^{\,(-\emptyword)}(m,\text{frc})\subseteq\textsf{PCD}^{\,(-\emptyword)}(m)$.
\end{lemma}

\begin{proof}[Sketch]
Let the frcCDGS $G=(N,\Sigma,(P_1,F_1),\dots,(P_n,F_n),S)$ be given. Let $X_f\notin N$.
To each $F_i$, associate the component $P_{\text{fail},i}=\{X\to X,X\to X_f\mid X\in F_i\}$.\footnote{We can omit these production sets if they are empty, but this is of no concern here.}
Consider the PCDGS $$G'=(N\cup\{X_f\},P_1,\dots,P_n,P_{\text{fail},1},\dots,P_{\text{fail},n},>,S)\,,$$
where $>$ is defined by  $P_{\text{fail},i}>P_i$ for all $i\in [n]$.
It should be obvious that the derivation relation defined by $G$ and by $G'$ are identical (formally ignoring the failure symbol $X_f$).\qed 
\end{proof}

\begin{theorem}\label{thm:pcd=cd}
For all derivation modes $m\in\{\leq k,=k \mid k\geq 1\}\linebreak[3]\cup\{*,t,\geq 1\}$, we find
 $\textsf{CD}^{\,(-\emptyword)}(m,\text{frc})=\textsf{PCD}^{\,(-\emptyword)}(m)$.
\end{theorem}

\begin{proof}[Sketch] For the inclusion $\subseteq$, we refer to Lemma~\ref{lem:cd<pcd}. We now discuss the inclusion $\supseteq$.
    For the  $=k$-modes, this follows via inclusion chain from Theorem~4 from~\cite{MitPauRoz94} and the Theorems~\ref{thm:GCac-characterization-by-OCD-all-modes},~\ref{thm:forbiddenCD=orderedCD} and~\ref{thm:frc=k}. 
    For the $t$-mode, this follows from Theorem~2 from~\cite{MitPauRoz94} and Theorem~\ref{thm:frc-CD}.
For the remaining modes $m\in\{\leq k \mid k\geq 1\}\linebreak[3]\cup\{*,=1\}$, we argue as follows.
If $P_i>P_j$, then we put all left-hand sides of productions in $P_i$ into the forbidden set $F_j$ of $P_j$, because applicability of $P_i$ means that any of its rules might be applicable for any of these modes. 
\qed
\end{proof}

With that, we have answered the open problem from \cite{MitPauRoz94} whether $\textsf{PCD}^{(-\emptyword)}$ is strictly included in $\textsf{MAT}_{ac}^{\,(-\emptyword)}$ for $=k$- and $\leq k$-mode.

\begin{corollary}\label{cor:CD=fRC=PCD}
    $\textsf{CD}^{\,(-\emptyword)}(\leq k, \text{frc}) = \textsf{fRC}^{\,(-\emptyword)} = \textsf{PCD}^{\,(-\emptyword)}(\leq k)$.
\end{corollary}
\begin{proof}
    This follows per \textit{Ringschluss} from Lemma~\ref{lem:CDfrc-in-fRCCD}, Theorem~\ref{thm:frccd1=frc}, Theorem~8 from~\cite{May72}, Theorem 3 from~\cite{MitPauRoz94}, and Theorem~\ref{thm:pcd=cd}.\qed
\end{proof}
\begin{toappendix}
\begin{remark}
    For the $t$-mode case of Lemma~\ref{thm:pcd=cd}, we can also give a direct simulation of PCDGS via CDGSfrc that is possibly of some independent interest.

For the $t$-mode, in a preparatory step, consider each $P_i$ and define $N_i$ as the set of all left-hand sides and $\Sigma_i$ as the set of all symbols that do occur in rules of $P_i$ but not on any left-hand side. Let $G_i(A)=(N_i,\Sigma_i,P_i,A)$ for any $A\in N_i$ and define for this context-free grammar $\cL_i(A)=\cL(G_i(A))$. Now observe that component $P_i$ will prevent $P_j$ from being applied if and only if the current sentential form $w$ contains any symbol $A\in N_i$ such that $\cL_i(A)\neq\emptyset$. (For context-free grammars, one can efficiently decide non-emptiness.) Hence, in the case of the $t$-mode, we only add those $A\in N_i$ with $\cL_i(A)\neq\emptyset$ into $F_j$.
\qed  
\end{remark}

\end{toappendix}

\LV{
This clarifies most relations between the language classes discussed in the paper for most modes, as can be also seen in Figure~\ref{fig:survey-of-results}. However, there are still some open questions concerning the $\geq k$-mode, which we discuss now.

 \begin{lemma}\label{lem:cd-in-cdfrc}
     For $k\geq 1$, $\textsf{CD}^{\,(-\emptyword)}(\geq k) \subset \textsf{CD}^{\,(-\emptyword)}(\geq k,\text{frc})$.
 \end{lemma}
 \begin{proof}
     The inclusion of $\textsf{CD}^{\,(-\emptyword)}(\geq k)$ in $\textsf{CD}^{\,(-\emptyword)}(\geq k,\text{frc})$ is obvious. 
     From \cite{HauJan94}, we know that $L=\{a^{2^n}\}$ is not in $\textsf{MAT}^{\,(-\emptyword)}$. Since $\textsf{CD}^{\,(-\emptyword)}(\geq k) \subseteq \textsf{MAT}^{\,(-\emptyword)}$ (see~\cite{MitPauRoz94}), we can conclude $L \notin \textsf{CD}^{(-\emptyword)}(\geq k)$. 
     
     Consider the frcCDGS~$G=(N,\Sigma, (P_1,F_1),(P_2,F_2),(P_3,F_3),A)$ where $P_1=\{A\to B,A \to C, A\to A \}$ with $F_1=\{B,C\}$, $P_2=\{B\to AA, B\to B\}$ with $F_2=\{A,C\}$, and $P_3=\{C\to a, C\to C\}$ with $F_3=\{A,B\}$. It is easily seen that $L(G,\geq k, \text{frc})=L$ and therefore $L \in \textsf{CD}^{\,(-\emptyword)}(\geq k,\text{frc})$. With that, the proposed inclusion is strict.\qed
 \end{proof}
 }

\LV{
\begin{lemma}\label{lem:cdfrc-k=2}
    For all $k\geq 2$, $\textsf{CD}^{\,(-\emptyword)}(\geq k, \text{frc})=\textsf{CD}^{\,(-\emptyword)}(\geq 2, \text{frc})$.
\end{lemma}
\begin{proof}
 Since $\textsf{CD}^{(-\emptyword)}(\geq k, \text{frc}) \supseteq \textsf{CD}^{(-\emptyword)}(\geq 2, \text{frc})$ is obviously true, it is sufficient to show the inclusion of the other direction. \\
    Let $G=(N,\Sigma, (P_1,F_1),\dots,(P_n,F_n),S)$ be an frcCDGS working in $\geq k$-mode with components $G_i=(N,\Sigma,P_i, S)$, for $i\in[n]$. Construct an frcCDGS\begin{align*}
				\tilde{G}=(\,&N \cup Y_\cup,\Sigma, (P_0,F_0), (P_{\text{end}},F_{\text{end}}),\\ &((P_{\text{pick},i},F_{\text{pick},i}),
				(P_{i,0},F_{i,0}),\dots,(P_{i,k+1},F_{i,k+1}))_{i\in[n]},S'\,)\,
	\end{align*}
    where $Y_\cup$ denotes $ \{Y\} \cup \{Y_{i,j}\}_{i\in [n], j\in [0\dots k]}$.
    For each component~$P_i$ of~$G$ add the following components to~$\tilde{G}$:
    \begin{itemize}
        \item $P_{\text{pick},i}=\{Y\to Y', Y'\to Y_i\} 
        $ 
        with $F_{\text{pick},i}=F_i \cup Y_\cup -Y $.
        With this, the $i$-th component is picked to be applied next. The forbidden context~$F_i$ has to be fulfilled to enter this component.
        
        \item $P_{i,0}=\{Y_i\to Y_{i,1}\} \cup 
        \{A \to w \mid A\to w \in P_i \}$ with $F_{i,0}=F_i \cup Y_\cup - Y_i $. Since the forbidden context~$F_i$ of component~$P_i$ has already been checked in $P_{\text{pick},i}$, it is not necessary to check it again when entering $P_{i,0}$.
        The nonterminal~$Y_{i,j}$ denotes that rules from component~$G_i$ are currently applied and that~$j$ derivation steps have already been taken. \\ After the component is left after the rule $Y_i\to Y_{i,1}$ was applied, it is blocked to be reentered as $Y_{i,1}\in F_{i,0}$. If the rule $Y_i\to Y_{i,1}$ was not applied, the only component that can be entered is $P_{i,0}$. The reentering of this component would not affect the objective of performing at least $k$ steps with $P_i$. 
        
        \item for $l\in [k-1]$: $P_{i,l}=\{Y_{i,l}\to Y_{i,l+1}\} \cup \{A\to w \mid A\to w \in P_i\}$ with $F_{i,l}= Y_\cup - Y_{i,l}$. These components are used to count that at least $k$ steps are performed. Like component~$P_{i,0}$, these components can be reentered when the counter~$Y_{i,l}$ is not increased. 
        
        \item $P_{i,k}=\{Y_{i,k}\to Y_{i,k}\} \cup \{A\to w \mid A\to w \in P_i\}$ with $F_{i,k}= Y_\cup - Y_{i,k}$. Now, at least $k$ steps according to~$P_i$ have been taken. This component can be reentered an arbitrary amount of times to perform $\geq k$ steps. 
        
        \item  $P_{i,k+1}=\{Y_{i,k}\to Y\} \cup \{(A\to w \mid A\to w \in P_i\}$ with $F_{i,k+1}= Y_\cup - Y_{i,k}$. With this, the derivation with rules from~$P_i$ ends. The counter~$Y_{i,k}$ is reset to~$Y$. Afterwards, the next component~$P_j$ can be picked by applying~$P_{\text{pick},j}$.
        
    \end{itemize}
    
    Further, add components $P_0=\{S'\to YS\}$ with $F_0=\emptyset$ and $P_{\text{end}}=\{Y\to \lambda\}$ with $F_{\text{end}}=N \cup Y_\cup - Y$ to start and end the computation, respectively.
\qed
\end{proof}
}

\LV{
\begin{corollary}
    For all $k\geq 2$, we find: $\textsf{CD}^{\,(-\emptyword)}(\geq k) \subseteq \textsf{CD}^{\,(-\emptyword)}(\geq k+1) \subset \textsf{CD}^{\,(-\emptyword)}(\geq k, \text{frc})=\textsf{CD}^{\,(-\emptyword)}(\geq k+1, \text{frc}) \subseteq \textsf{PCD}^{\,(-\emptyword)}(\geq k) \subseteq \textsf{PCD}^{\,(-\emptyword)}(\geq k+1) $.
\end{corollary}
\begin{proof}
    The inclusion $\textsf{CD}^{\,(-\emptyword)}(\geq k) \subseteq \textsf{CD}^{\,(-\emptyword)}(\geq k+1)$ was proved in~\cite{Csuetal94all}. The strict inclusion of $\textsf{CD}^{\,(-\emptyword)}(\geq k+1)$ in $\textsf{CD}^{\,(-\emptyword)}(\geq k+1, \text{frc})$ and the equality of $\textsf{CD}^{\,(-\emptyword)}(\geq k, \text{frc})$ and $ \textsf{CD}^{\,(-\emptyword)}(\geq k+1, \text{frc})$ were shown in Lemma~\ref{lem:cdfrc-k=2} and Lemma~\ref{lem:cd-in-cdfrc} respectively. Lastly, $\textsf{PCD}^{\,(-\emptyword)}(\geq k) \subseteq \textsf{PCD}^{\,(-\emptyword)}(\geq k+1)$ obviously holds in the same way as in the non-prioritized case. \qed
\end{proof}
}

\section{Summary and Open Problems}
\label{sec:conclusions}

	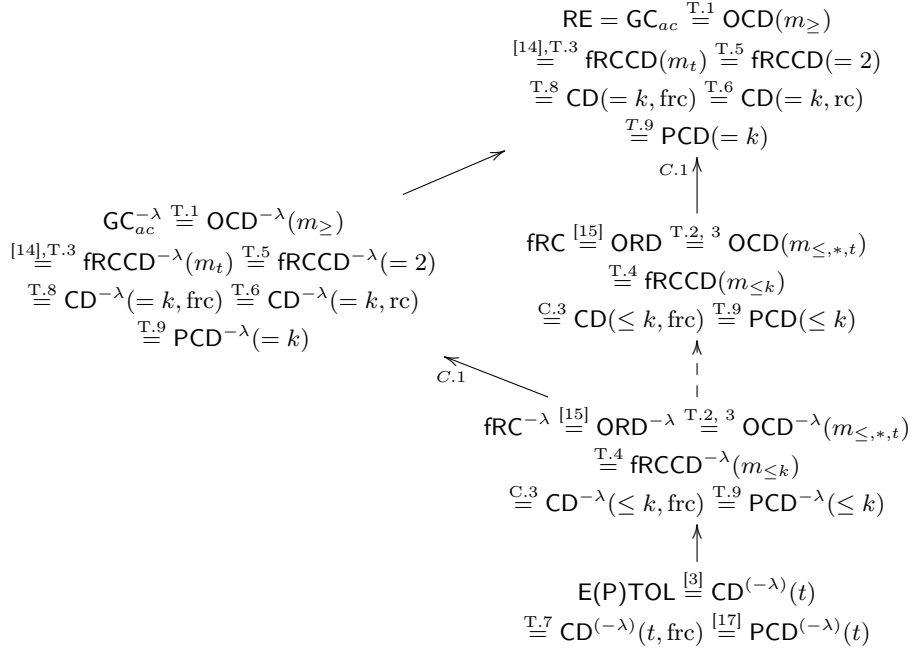
\begin{figure}[tb]
		\xymatrix@R=1.5em@ C=1.25em{
			& \txt{$\textsf{RE}
            =\textsf{GC}_{ac}\overset{\text{T.\ref{thm:GCac-characterization-by-OCD-all-modes}}}{=}\textsf{OCD}(m_{\geq})$\\
            $\overset{\text{\cite{Mas2009},T.\ref{thm:forbiddenCD=orderedCD}}}{=}\textsf{fRCCD}(m_t)
            \overset{\text{T.\ref{thm:mode-hierarchy-collaps}}}{=}\textsf{fRCCD}(=2)$\\
            $\overset{\text{T.\ref{thm:frc=k}}}{=}\textsf{CD}(=k,\text{frc})
            \overset{\text{T.\ref{thm:CsuDasKelPau-book}}}{=}\textsf{CD}(=k,\text{rc})$\\$\overset{T.\ref{thm:pcd=cd}}{=}\textsf{PCD}(=k)$}    \\
			\ar[ur] \txt{$\textsf{GC}_{ac}^{-\emptyword} 
            \overset{\text{T.\ref{thm:GCac-characterization-by-OCD-all-modes}}}{=}\textsf{OCD}^{-\emptyword}(m_{\geq})$\\
            $\overset{\text{\cite{Mas2009},T.\ref{thm:forbiddenCD=orderedCD}}}{=}\textsf{fRCCD}^{-\emptyword}(m_t) 
            \overset{\text{T.\ref{thm:mode-hierarchy-collaps}}}{=} \textsf{fRCCD}^{-\emptyword}(=2)$\\
            $\overset{\text{T.\ref{thm:frc=k}}}{=}\textsf{CD}^{-\emptyword}(=k,\text{frc})
            \overset{\text{T.\ref{thm:CsuDasKelPau-book}}}{=}\textsf{CD}^{-\emptyword}(=k,\text{rc})$\\
            $\overset{\text{T.\ref{thm:pcd=cd}}}{=}\textsf{PCD}^{-\emptyword}(=k)$}
            &  \ar[u]^{C.\ref{cor:frccd1<frccd2}} \txt{$\textsf{fRC} 
            \overset{\text{\cite{May72}}}{=} \textsf{ORD}
            \overset{\text{T.\ref{thm:ordered-t}, \ref{thm:forbiddenCD=orderedCD}}}{=} \textsf{OCD}(m_{\leq,*,t}) $ \\ 
            $ \overset{\text{T.\ref{thm:frccd1=frc}}}{=}\textsf{fRCCD}(m_{\leq k} ) $ \\ 
            $\overset{\text{C.\ref{cor:CD=fRC=PCD}}}{=} \textsf{CD}(\leq k, \text{frc}) 
            \overset{\text{T.\ref{thm:pcd=cd}}}{=} \textsf{PCD}(\leq k)$}  \\
			&  \ar[ul]^{C.\ref{cor:frccd1<frccd2}} \ar@{-->}[u] \txt{$\textsf{fRC}^{-\emptyword} 
            \overset{\text{\cite{May72}}}{=} \textsf{ORD}^{-\emptyword}
            \overset{\text{T.\ref{thm:ordered-t}, \ref{thm:forbiddenCD=orderedCD}}}{=} \textsf{OCD}^{-\emptyword}(m_{\leq, *,t}) $\\
            $\overset{\text{T.\ref{thm:frccd1=frc}}}{=}\textsf{fRCCD}^{-\emptyword}(m_{\leq k} )$ \\ 
            $\overset{\text{C.\ref{cor:CD=fRC=PCD}}}{=} \textsf{CD}^{-\emptyword}(\leq k, \text{frc}) 
            \overset{\text{T.\ref{thm:pcd=cd}}}{=} \textsf{PCD}^{-\emptyword}(\leq k)$}  \\
			& \ar[u] \txt{$\textsf{E(P)TOL} 
            \overset{\text{\cite{Csu94}}}{=}\textsf{CD}^{(-\emptyword)}(t)$\\ 
            $\overset{\text{T.\ref{thm:frc-CD}}}{=} \textsf{CD}^{(-\emptyword)}(t,\text{frc})
            \overset{\text{\cite{MitPauRoz94}}}{=}\textsf{PCD}^{(-\emptyword)}(t)$}
		}
    \caption{Dashed arrows denote inclusions and solid arrows denote strict inclusions.\\
    Note that $m_{\geq}$ denotes any mode from $ \{=k, \geq k \mid k\geq 2\}$, 
	$m_t $ from $ \{=k, \geq k \mid k\geq 2\}\linebreak[3] \cup \{t\}$, 
    \mbox{$m_{\leq k} $ from $ \{\leq k \mid k \geq 1 \} \cup \{*, = 1, \geq 1\}$}, 
    and $m_{\leq,*,t}$ from $\{\leq k \mid k \geq 1 \} \cup \{*,t\}$.%
    }
    \label{fig:survey-of-results}
	\end{figure}

In this paper, we revisited regulated variants of CD grammar systems.
Although many these formalisms have been introduced over three decades ago, many relations between them have been left open until now. We clarify the inclusions for most of the possible combinations. This is illustrated in Figure~\ref{fig:survey-of-results}.

Still, several questions remain open. For instance, it is an open question for more than five decades now if (disregarding the possibility to generate the empty word) ordered grammars without erasing rules are strictly less powerful than those with possibly having erasing rules.
It is also unclear if (or which of) the following inclusions are strict:
$$\textsf{CD}^{(-\lambda)}(\geq k,\text{frc})\subseteq \textsf{PCD}^{(-\lambda)}(\geq k)\subseteq \textsf{fRCCD}^{(-\lambda)}(\geq k)(=\textsf{GC}_{ac}^{-\lambda})\,.$$

Furthermore, notice that we only looked at forbidden context conditions, but 
permitting random context conditions also deserve a separate exploration.
Finally, we would suggest looking at descriptional complexity parameters, like number of components, number of nonterminals, number of rules per component etc. A systematic study is lacking here.

\paragraph{Acknowledgement.}
All authors thank DST-DAAD project (DST/IND/DAAD/P-01-2024(G) resp. DAAD PPP 57724085) for supporting the work carried out under the scope of the project.\SV{\todo{If UCNC does not allow for more space, we should cut out something from the conclusions.}}
\bibliographystyle{plain}
\bibliography{ab,rgs}
\end{document}